\documentclass{article}
\usepackage{spconf,amsmath,amsfonts,graphicx}
\usepackage{tikz} 
\usepackage[utf8]{inputenc}
\usepackage{pgfplots} 
\usepackage{amsthm}
\usepackage{amssymb}

\newtheorem{proposition}{Proposition}

\newtheorem{remark}{Remark}
\usepackage{tabu,longtable}
\usepackage{diagbox}
\usepackage{threeparttable}
\usepackage{makecell}
\usepackage{booktabs}
\usepackage{color}
\title{Beamforming Design and Performance Evaluation for RIS-aided Localization using LEO Satellite Signals}

\name{Lei~Wang, Pinjun~Zheng, Xing~Liu, Tarig~Ballal, Tareq~Y.~Al-Naffouri}
\address{King Abdullah University of Science and Technology, Thuwal, Saudi Arabia.}

\begin{document}
%\ninept
%
\maketitle
\begin{abstract}
The growing availability of low-Earth orbit (LEO) satellites, coupled with the anticipated widespread deployment of reconfigurable intelligent surfaces (RISs), opens up promising prospects for new localization paradigms. This paper studies RIS-aided localization using LEO satellite signals. The Cram\'er-Rao bound of the considered localization problem is derived, based on which an optimal RIS beamforming design that minimizes the derived bound is proposed. Numerical results demonstrate the superiority of the proposed beamforming scheme over benchmark alternatives, while also revealing that the synergy between LEO satellites and RISs holds the promise of achieving localization accuracy at the meter or even sub-meter level.

\end{abstract}
\section{Introduction}
\label{sec:intro}
In recent years, low-Earth orbit (LEO) satellites, which are typically deployed at altitudes from 500 to 2000 km~\cite{Li2022LEO}, have received considerable attention. Although existing LEO constellations were not originally intended for localization, recent studies have increasingly recognized the capability of LEO satellite localization, either through the dedicated satellite systems or by utilizing signals of opportunity~\cite{kassas2019new,reid2018broadband}. LEO satellites show the potential to be a complement or an alternative to the global navigation satellite systems (GNSS) that reside in medium-Earth orbit (MEO), thanks to their desirable attributes such as stronger received signal power, more visible satellites, and higher frequency diversity, etc~\cite{Kozhaya2023Bilnd,dureppagari2023ntn}.

Within the existing literature, numerous works have delved into integrating terrestrial and non-terrestrial localization systems~\cite{Zheng20235G,Zheng2023Attitude}.
As an emerging technique, reconfigurable intelligent surfaces (RISs) bring new opportunities for both terrestrial and non-terrestrial networks-based localization~\cite{Bjornson2022Reconfigurable,Tekb2022Reconfigurable}. In terrestrial wireless systems, it has been shown that RIS can improve localization accuracy by reshaping the propagation environment, providing additional location references, and delivering more measurements~\cite{ Henk2022Beamforming,Zheng2023JrCUP,Chen2023Multi}. Nonetheless, the non-terrestrial network-based localization with RIS involved is still in its early stages. 

Current LEO satellite signals primarily operate in the Ku/Ka band, closely aligned with the mmWave frequencies used in terrestrial 5G communications. This similarity inspires our research into leveraging ground RISs to enhance LEO satellite localization. The nature of LEO satellite signals, such as long propagation distance, high Doppler shifts, and abundant satellite resources, distinguishes them from the terrestrial systems~\cite{Rui2022LEO}. In this paper, we first derive the Cram\'er-Rao bound (CRB) for RIS-aided LEO satellite localization, based on which a localization-oriented RIS beamforming design is proposed, which is shown to outperform existing alternatives. We also compare the performance of LEO satellite localization to that of terrestrial base stations (BSs), revealing the potential of LEO satellite localization.

\section{System Model}
As shown in Fig.~\ref{fig1}, we consider an RIS-aided downlink localization system with a single-antenna user equipment (UE) at unknown location $\mathbf{p}_{\mathrm{u}}\in \mathbb{R}^{3\times 1}$, an RIS at known location $\mathbf{p}_{\mathrm{r}}\in \mathbb{R}^{3\times 1}$, and a LEO satellite with a known position $\mathbf{p}_{\mathrm{s}}\in \mathbb{R}^{3\times 1}$ and velocity $\boldsymbol{v}\in \mathbb{R}^{3\times 1}$. The orientations of the LEO satellite and RIS are assumed to be known.
Besides, the satellite and RIS are equipped with uniform planar arrays (UPAs) of $N_\mathrm{s} = N_\mathrm{s}^x N_\mathrm{s}^y$ and $N_{\mathrm{r}} = N_{\mathrm{r}}^x N_{\mathrm{r}}^y$ elements, respectively. $N_\mathrm{s}^x$ and $N_\mathrm{s}^y$ represent the numbers of elements on the
x- and y-axes of the UPA on LEO satellite, respectively. Similarly, $N_{\mathrm{r}}^x$ and  $N_{\mathrm{r}}^y$ define the shape of the RIS UPA. The element spacings of the two UPAs are half wavelength. The UE locates itself by exploiting the received signals from the direct satellite-UE path and cascaded satellite-RIS-UE path. This is a single satellite and single RIS scenario. The cases with multiple satellites/RISs will be discussed in Section~\ref{sec:multLEO}.

\begin{figure}[t]
  \centering
  \includegraphics[width=2.in]{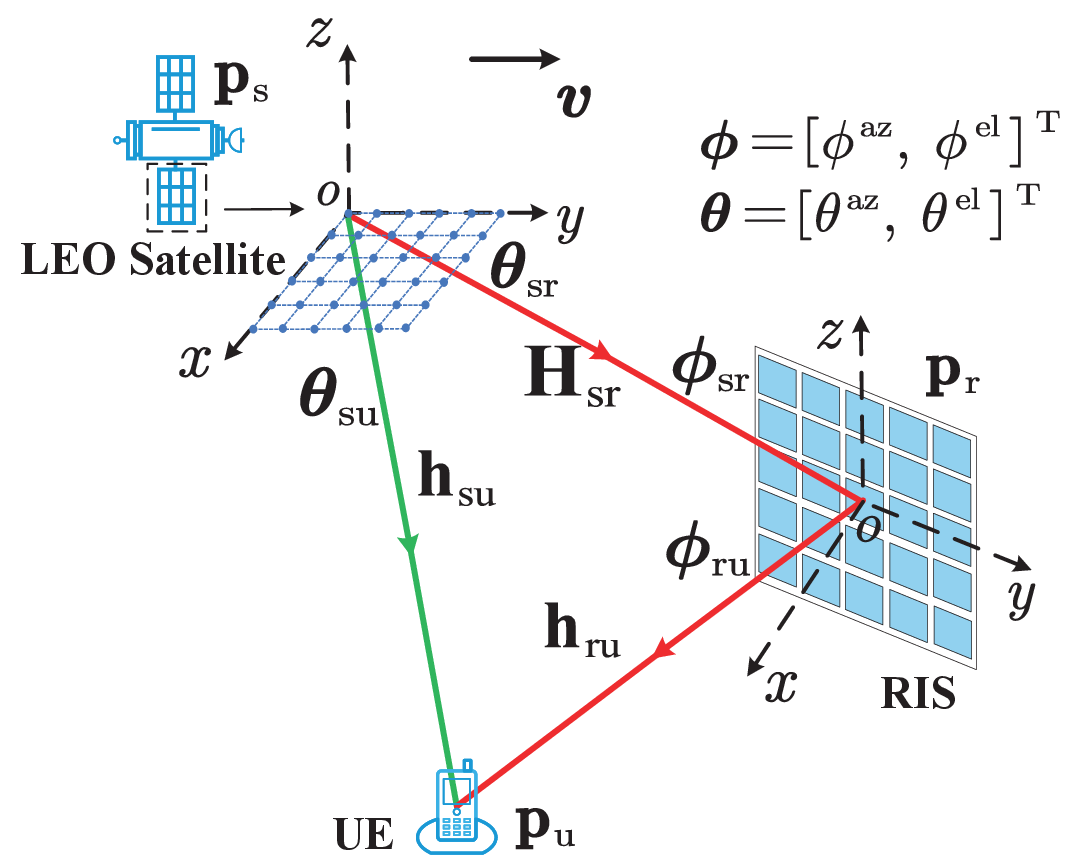}
  \caption{The considered setup with LEO satellite, RIS and UE.}
  \label{fig1}
\end{figure}

\subsection{Channel Model}
The considered system is depicted in Fig.~\ref{fig1}, where the downlink channel can be represented by
\begin{equation}
\mathbf{h}^\mathrm{T}(t,f)=\mathbf{h}^\mathrm{T}_{\mathrm{su}}(t,f)+\mathbf{h}^\mathrm{T}_{\mathrm{ru}}(f)\mathbf{\Omega}\mathbf{H}_{\mathrm{sr}}(t,f),
\end{equation}
where $\mathbf{h}_{\mathrm{su}}(t,f) \in \mathbb{C}^{N_{\mathrm{s}} \times 1}$ denotes the channel from the satellite to the UE, $\mathbf{H}_{\mathrm{sr}}(t,f) \in \mathbb{C}^{N_{\mathrm{r}} \times N_{\mathrm{s}}}$ represents the channel from the satellite to the RIS, $\mathbf{h}_{\mathrm{ru}}(f) \in \mathbb{C}^{N_{\mathrm{r}} \times 1}$ is the channel from the RIS to the UE, and $\mathbf{\Omega} =\mathrm{diag}(\boldsymbol{\omega})\in \mathbb{C}^{N_{\mathrm{r}} \times N_{\mathrm{r}}}$ denotes RIS reflection matrix with $\boldsymbol{\omega} = [\omega_1,\ldots,\omega_{N_\mathrm{r}}]^\mathrm{T}\in \mathbb{C}^{N_{\mathrm{r}} \times 1}$. 
Here, $\omega_{n_\mathrm{r}}$ denotes the reflection coefficient of the $n_\mathrm{r}$-th RIS element, which can be characterized as $\omega_{n_r}=\frac{Z_{n_r}-Z_0}{Z_{n_r}+Z_0}$ with 
$Z_{n_r}$ being the reconfigurable impedance and $Z_0$ being the reference impedance (typically $Z_0=50~\Omega$)~\cite{Shen2022RIS}.\footnote{Here we consider a single connected reconfigurable impedance network.} For a passive RIS setup, we have $\Re(Z_{n_r})\geq 0$, which follows the constraint 
\begin{equation}\label{eq:omegaconstraint}
    \|\boldsymbol{\omega}\|_\infty\leq 1.
\end{equation}
Given the lack of a shared central processing unit (CPU) between LEO satellites and ground devices, we assume that the RIS and LEO satellites operate in a non-cooperative mode. Thus, unlike the terrestrial scenario~\cite{Henk2022Beamforming}, $\boldsymbol{\Omega}$ remains unchanged throughout the LEO satellite signal transmissions.

The satellite-UE channel can be represented by~\cite{You2020Model} \footnote{For simplicity, we ignore the multipath effect, which will be the subject of future exploration.}
\begin{equation}
\mathbf{h}_{\mathrm{su}}(t,f)=\alpha_{\mathrm{su}}e^{j2\pi(t\nu_{\mathrm{su}}-f\tau_{\mathrm{su}})}\mathbf{a}_{\mathrm{s}}(\boldsymbol{\theta}_{\mathrm{su}}),
\end{equation}
where $\alpha_{\mathrm{su}}$, $\nu_{\mathrm{su}}$, and $\tau_{\mathrm{su}}$ denote the complex-valued channel gain, the Doppler shift, the propagation delay from the satellite to the UE, respectively.
$\boldsymbol{\theta}_{\mathrm{su}}=\left[\theta_{\mathrm{su}}^{\mathrm{az}}, \theta_{\mathrm{su}}^{\mathrm{e l}}\right]^{\mathrm{T}}$ represents the angle-of-departure (AoD) for the satellite to the UE, including two components, i.e., the azimuth $\theta_{\mathrm{su}}^{\mathrm{az}}$ and elevation $\theta_{\mathrm{su}}^{\mathrm{el}}$.
The array response vector of the satellite is formulated as
$\mathbf{a}_{\mathrm{s}}(\boldsymbol{\theta})=e^{-j \frac{2 \pi f_c}{c} \mathbf{P}_{\mathrm{s}}^{\mathrm{T}} \mathbf{t}(\boldsymbol{\theta})}$,
where
$\mathbf{P}_{\mathrm{s}}=\left[\mathbf{p}_{\mathrm{s}, 1}, \cdots, \mathbf{p}_{\mathrm{s}, N_\mathrm{s}}\right], \mathbf{t}(\boldsymbol{\theta})=\left[\mathcal{C}_{\theta^{\mathrm{az}}} \mathcal{C}_{\theta^{\mathrm{el}}}, \mathcal{S}_{\theta^{\mathrm{az}}} \mathcal{C}_{\theta^{\mathrm{el}}}, \mathcal{S}_{ \theta^{\mathrm{el}}}\right]^{\mathrm{T}}$, 
$f_c$ is the carrier frequency, $c$ is the speed of light, $\mathbf{p}_{\mathrm{s}, i}$ is the coordinate of the $i$-th antenna in the body coordinate system (BCS), and $\mathcal{C}_\theta$ and $\mathcal{S}_\theta$ denote $\cos{\theta}$ and $\sin{\theta}$, respectively.

The satellite-RIS channel can be represented by~\cite{You2020Model}
\begin{equation}
\mathbf{H}_{\mathrm{sr}}(t,f)=\alpha_{\mathrm{sr}}e^{j2\pi(t\nu_{\mathrm{sr}}-f\tau_{\mathrm{sr}})}\mathbf{a}_{\mathrm{r}}(\boldsymbol{\phi}_{\mathrm{sr}})\mathbf{a}^{\mathrm{T}}_{\mathrm{s}}(\boldsymbol{\theta}_{\mathrm{sr}}),
\end{equation}
where $\alpha_{\mathrm{sr}}$, $\nu_{\mathrm{sr}}$, and $\tau_{\mathrm{sr}}$ respectively denote the complex-valued channel gain, the Doppler shift, and the propagation delay from the satellite to the RIS, $\boldsymbol{\theta}_{\mathrm{sr}} = \left[\theta_{\mathrm{sr}}^{\mathrm{az}}, \theta_{\mathrm{sr}}^{\mathrm{el}}\right]^{\mathrm{T}}$ and $\boldsymbol{\phi}_{\mathrm{sr}} = \left[\phi_{\mathrm{sr}}^{\mathrm{az}}, \phi_{\mathrm{s r}}^{\mathrm{el}}\right]^{\mathrm{T}}$ respectively denote the AoD and angle-of-arrival (AoA) from the satellite to the RIS, and $\mathbf{a}_{\mathrm{r}}(\boldsymbol{\phi})=$ denotes the array response vector of the RIS that can be written as
$\mathbf{a}_{\mathrm{r}}(\boldsymbol{\phi})=e^{-j \frac{2 \pi f_c}{c}\mathbf{P}_{\mathrm{r}}^{\mathrm{T}} \mathbf{t}(\boldsymbol{\phi})}$,
with $\mathbf{P}_{\mathrm{r}} \in \mathbb{R}^{3 \times N_{\mathrm{r}}}$ the RIS elements' coordinates in its BCS.

Finally, assuming the RIS and UE are stationary over the transmission period, the channel between them can be expressed as~\cite{Henk2022Beamforming}
\begin{equation}
\mathbf{h}^\mathrm{T}_{\mathrm{ru}}(f)=\alpha_{\mathrm{ru}}e^{-j2\pi f\tau_{\mathrm{ru}}}\mathbf{a}^{\mathrm{T}}_{\mathrm{r}}(\boldsymbol{\phi}_{\mathrm{ru}}),
\end{equation}
where $\alpha_{\mathrm{ru}}$, $\tau_{\mathrm{ru}}$, and $\boldsymbol{\phi}_{\mathrm{ru}} =\left[\phi_{\mathrm{ru}}^{\mathrm{az}}, \phi_{\mathrm{ru}}^{\mathrm{el}}\right]^{\mathrm{T}}$ denote the complex-valued channel gain, the propagation delay from the RIS to the UE, and the AoD of the RIS, respectively.

Based on the geometric relationships shown in Fig.~\ref{fig1}, the propagation delays and Doppler shifts can be expressed as
$\tau_\mathrm{su} = \frac{\|\mathbf{p}^{\mathrm{s}} - \mathbf{p}^{\mathrm{u}}\|}{c}+\Delta,\tau_\mathrm{sru} = \tau_\mathrm{sr}+\tau_\mathrm{ru}=\frac{\|\mathbf{p}^{\mathrm{s}} - \mathbf{p}^{\mathrm{r}}\|+\|\mathbf{p}^{\mathrm{r}} - \mathbf{p}^{\mathrm{u}}\|}{c}+\Delta,
\nu_{\mathrm{su}} = \frac{\boldsymbol{v}^\mathrm{T}(\mathbf{p}^{\mathrm{u}}-\mathbf{p}^{\mathrm{s}})}{\lambda\|\mathbf{p}^{\mathrm{u}}-\mathbf{p}^{\mathrm{s}}\|},
\nu_{\mathrm{sr}} = \frac{\boldsymbol{v}^\mathrm{T}(\mathbf{p}^{\mathrm{r}}-\mathbf{p}^{\mathrm{s}})}{\lambda\|\mathbf{p}^{\mathrm{r}}-\mathbf{p}^{\mathrm{s}}\|}$,
where $\Delta$ is the unknown constant clock offset (without drift) between the satellite and UE. Please refer to, e.g.,~\cite {Zheng2023JrCUP,Zheng2022Coverage}, for the details regarding the definitions of the aforementioned angles.

\subsection{Signal Model}
We adopt orthogonal frequency division multiplexing (OFDM) for signal modulation since it is widely used in practical systems (e.g., Starlink)~\cite{Kozhaya2023Bilnd,Neinavaie2023OFDM}.
Suppose that the satellite sequentially conducts $M$ OFDM transmissions over $N$ subcarriers. The structure of the satellite signal can be acquired by using the methods in~\cite{Kozhaya2023Bilnd,Neinavaie2023OFDM}. 
For the $m$-th transmission, a multi-carrier OFDM signal $\mathbf{s}_m=[s_m[1],\ldots,s_m[N]]^{\mathrm{T}} \in \mathbb{C}^{N \times 1}$ for $m=1,\ldots,M$, is precoded by a vector $\mathbf{f}_m \in \mathbb{C}^{N_{\mathrm{s}} \times 1}$.
Here, we use an analog precoder that each entry of $\mathbf{f}_m$ (denoted as $f_{m,i}$, $i=1,\dots,N_\mathrm{s}$) is constrained to satisfy $f_{m,i}f_{m,i}^*=1/N_\mathrm{s}$~\cite{Ayach2014Spatially}.

The downlink received signal at the UE side over the $n$-th subcarrier of the $m$-th transmission can be expressed as
\begin{align}\label{eq_received_signal}
x_m[n]=\sqrt{P} \mathbf{h}^{\mathrm{T}}_{m}[n] \mathbf{f}_m s_m[n] + v_m[n],\;n=1,\ldots,N,
\end{align}
where $P$ is the transmitted power, $v_m[n]\sim \mathcal{CN}(0,\sigma^2)$ is circularly symmetric complex Gaussian noise, and $\mathbf{h}_{m}[n]$ denotes the channel vector on the $n$-th subcarrier for the $m$-th transmission, written as $\mathbf{h}_m[n] = \mathbf{h}(mT,n\Delta f)$, with $T$ being the transmission period, and $\Delta f$ being the subcarrier spacing.

\subsection{Problem Formulation}
Our problems of interest are to: ($\romannumeral+1$) derive Cram\'er-Rao bound (CRB) for the estimation of the UE position $\mathbf{p}_{\mathrm{u}}$ from the observations $\{x_m[n]\}_{\forall{m,n}}$ in~(\ref{eq_received_signal}) by employing a transformation of parameters from the unknown channel-domain parameter vector $\boldsymbol{\gamma} = [\boldsymbol{\chi}^{\mathrm{T}}_\mathrm{c},\boldsymbol{\chi}^{\mathrm{T}}_\mathrm{n}]^{\mathrm{T}} \in \mathbb{R}^{11\times 1}$ to the unknown location-domain parameter vector $\boldsymbol{\eta} = [\boldsymbol{\chi}^{\mathrm{T}}_\mathrm{p},\boldsymbol{\chi}^{\mathrm{T}}_\mathrm{n}]^{\mathrm{T}} \in \mathbb{R}^{8 \times 1}$. Here $\boldsymbol{\chi}_\mathrm{c}=[\tau_\mathrm{su}, \boldsymbol{\theta}_\mathrm{su}^{\mathrm{T}}, \tau_\mathrm{sru}, \boldsymbol{\phi}_\mathrm{ru}^{\mathrm{T}}, \nu_\mathrm{su}]^{\mathrm{T}} \in \mathbb{R}^{7\times 1}$, $\boldsymbol{\chi}_\mathrm{n}=[\Re(\alpha_\mathrm{su}), \Im(\alpha_\mathrm{su}), \Re(\alpha_\mathrm{sru}), \Im(\alpha_\mathrm{sru})]^{\mathrm{T}} \in \mathbb{R}^{4\times 1}$ with $\alpha_\mathrm{sru}=\alpha_\mathrm{sr}\alpha_\mathrm{ru}$, and $\boldsymbol{\chi}_\mathrm{p}=[\mathbf{p}_\mathrm{u}^\mathrm{T}, \Delta]^{\mathrm{T}} \in \mathbb{R}^{4\times 1}$; ($\romannumeral+2$)~design an optimal RIS profile $\boldsymbol{\omega}$ that minimizes the derived CRB.

\section{Fisher Information Analysis}
\subsection{Single LEO Satellite \& Single RIS}\label{sec:singleLEOCRB}
First, we compute the Fisher information matrix (FIM) of the channel parameter vector $\boldsymbol{\gamma}$. Since $\{x_m[n]\}_{\forall{m,n}}$ follow complex Gaussian distribution, the corresponding FIM of $\boldsymbol{\gamma}$ can be expressed using the Slepian-Bangs formula~\cite{Kay1993Slepian} as
\begin{align}\label{eq:S-B}
\small \mathbf{J}_{\boldsymbol{\gamma}} =\frac{2}{\sigma^2}\sum^M_{m=1}\sum^N_{n=1}\Re\left\{\frac{\partial \mu_m[n]}{\partial \boldsymbol{\gamma} }\left(\frac{\partial \mu_m[n]}{\partial \boldsymbol{\gamma}}\right)^\mathrm{H} \right\} \in\mathbb{R}^{11\times 11},
\end{align}
where $\mu_m[n]\!=\!\sqrt{P} \mathbf{h}^{\mathrm{T}}_{m}[n] \mathbf{f}_m s_m[n]$ is the noise-free term of the received signal in~(\ref{eq_received_signal}). 
The partial derivatives in~(\ref{eq:S-B}) are computed~as:
\begin{equation}
\resizebox{.99\hsize}{!}{$
\begin{aligned}
&\frac{\partial\mu_m[n]}{\partial \tau_\mathrm{su}} = -j2\pi n\Delta f\sqrt{P}\alpha_\mathrm{su} e^{j2\pi(mT\nu_{\mathrm{su}} - n\Delta f\tau_{\mathrm{su}})} \beta(\boldsymbol{\theta}_\mathrm{su}), \notag \\
&\frac{\partial\mu_m[n]}{\partial \theta_\mathrm{su}^\star} = \sqrt{P}\alpha_\mathrm{su} e^{j2\pi(mT\nu_{\mathrm{su}} - n\Delta f\tau_{\mathrm{su}})} \frac{\partial{\mathbf{a}}^\mathrm{T}_\mathrm{s}(\boldsymbol{\theta}_\mathrm{su})}{\partial \theta_\mathrm{su}^\star}\mathbf{f}_m s_m[n],  \notag \\
&\frac{\partial\mu_m[n]}{\partial \tau_\mathrm{sru}} = -j2\pi n\Delta f\sqrt{P}\alpha_\mathrm{sru} e^{j2\pi(mT\nu_{\mathrm{sr}} - n\Delta f\tau_{\mathrm{sru}})} \mathbf{b}^\mathrm{T}_\mathrm{r}(\boldsymbol{\phi}_\mathrm{ru})\boldsymbol{\omega}\beta(\boldsymbol{\theta}_{\mathrm{sr}}) ,  \notag \\
&\frac{\partial\mu_m[n]}{\partial \phi_\mathrm{ru}^\star} = \sqrt{P}\alpha_\mathrm{sru} e^{j2\pi(mT\nu_{\mathrm{sr}} - n\Delta f\tau_{\mathrm{sru}})}  \frac{\partial {\mathbf{b}}^\mathrm{T}_\mathrm{r}(\boldsymbol{\phi}_\mathrm{ru})}{\partial \phi_\mathrm{ru}^\star}\boldsymbol{\omega}\beta(\boldsymbol{\theta}_{\mathrm{sr}}),  \notag \\
&\frac{\partial\mu_m[n]}{\partial \nu_\mathrm{su}} = j2\pi mT\sqrt{P}\alpha_\mathrm{su} e^{j2\pi(mT\nu_{\mathrm{su}} - n\Delta f\tau_{\mathrm{su}})} \beta(\boldsymbol{\theta}_{\mathrm{su}}), \notag\\
&\frac{\partial\mu_m[n]}{\partial \Lambda({\alpha_\mathrm{su}})} = l_{\Lambda}\sqrt{P} e^{j2\pi(mT\nu_{\mathrm{su}} - n\Delta f\tau_{\mathrm{su}})}  \beta(\boldsymbol{\theta}_{\mathrm{su}}), \notag \\
&\frac{\partial\mu_m[n]}{\partial \Lambda({\alpha_\mathrm{sru}})} = l_{\Lambda}\sqrt{P} e^{j2\pi(mT\nu_{\mathrm{sr}} - n\Delta f\tau_{\mathrm{sru}})}  \mathbf{b}^\mathrm{T}_\mathrm{r}(\boldsymbol{\phi}_\mathrm{ru})\boldsymbol{\omega}\beta(\boldsymbol{\theta}_{\mathrm{sr}}), \notag \\
\end{aligned}
$}
\end{equation}
where $\star \in \{\mathrm{az},\mathrm{el}\}$, $\Lambda \in \{\Re,\Im\}$, $l_{\Re} = 1$, $l_{\Im} = j$,
$\beta(\boldsymbol{\theta})= \mathbf{a}^\mathrm{T}_\mathrm{s}(\boldsymbol{\theta})\mathbf{f}_m s_m[n]$, and $\mathbf{b}_\mathrm{r}(\boldsymbol{\phi})=\mathbf{a}_\mathrm{r}(\boldsymbol{\phi})\odot \mathbf{a}_\mathrm{r}(\boldsymbol{\phi}_\mathrm{sr})$ with $\odot$ denoting the Hadamard product.

We then transform the FIM of the channel parameters $\boldsymbol{\gamma}$ to the location parameters $\boldsymbol{\eta}$ by using the transformation matrix $\mathbf{T}=\frac{\partial\boldsymbol{\gamma}^\mathrm{T}}{\partial \boldsymbol{\eta}}\in\mathbb{R}^{8\times 11}$, and the resulting FIM of $\boldsymbol{\eta}$ is given by $\mathbf{J}_{\boldsymbol{\eta}}=\mathbf{T}\mathbf{J}_{\boldsymbol{\gamma}}\mathbf{T}^\mathrm{T}\in\mathbb{R}^{8 \times 8}$. 
We adopt $\text{PEB} \triangleq \sqrt{\text{tr}\big(\big[\mathbf{J}_{\boldsymbol{\eta}}^{-1}\big]_{1:3,1:3}\big)}$ as the performance metric that evaluates the lower bound of the estimation root mean squared error (RMSE) of $\mathbf{p}_{\mathrm{u}}$.

\subsection{Multiple LEO Satellites \& Multiple RISs}
\label{sec:multLEO}
Suppose there are $K$ LEO satellites and $L$ RISs available for cooperative localization.
The channel-domain parameters related to the $k$-th LEO satellite can be summarized as $\boldsymbol{\gamma}_k=[\boldsymbol{\gamma}_{\mathrm{LoS},k}^\mathrm{T},\boldsymbol{\gamma}_{\mathrm{RIS},k,1}^\mathrm{T},\dots,\boldsymbol{\gamma}_{\mathrm{RIS},k,L}^\mathrm{T}]^\mathrm{T}$.
Here, $\boldsymbol{\gamma}_{\mathrm{LoS},k}$ contains the channel parameters of the channel between the $k$-th satellite and the UE (i.e., $\tau_\mathrm{su}, \boldsymbol{\theta}_\mathrm{su},  \nu_\mathrm{su},\Re(\alpha_\mathrm{su}),\Im(\alpha_\mathrm{su})$), while $\boldsymbol{\gamma}_{\mathrm{RIS},k,\ell}$ contains the channel parameters of the cascaded channel between the $k$-th satellite, the $\ell$-th RIS, and the UE (i.e., $\tau_\mathrm{sru},\boldsymbol{\phi}_\mathrm{ru},\Re(\alpha_\mathrm{sru}),\Im(\alpha_\mathrm{sru})$). 
Similarly, the FIM of the channel parameters related to the $k$-th satellite, $\mathbf{J}_{\boldsymbol{\gamma}_k}$, can be computed via~\eqref{eq:S-B}.
Then, the FIM of the total channel-domain parameters is given by $\mathbf{J}_{\boldsymbol{\gamma},\mathrm{tot}}=\mathrm{blkdiag}\left\{\mathbf{J}_{\boldsymbol{\gamma}_1},\mathbf{J}_{\boldsymbol{\gamma}_2},\dots,\mathbf{J}_{\boldsymbol{\gamma}_K}\right\}$ and the total transformation matrix is given by $\mathbf{T}_\mathrm{tot}=\big[\frac{\partial\boldsymbol{\gamma}_1^\mathrm{T}}{\partial \boldsymbol{\eta}},\frac{\partial\boldsymbol{\gamma}_2^\mathrm{T}}{\partial \boldsymbol{\eta}},\dots,\frac{\partial\boldsymbol{\gamma}_K^\mathrm{T}}{\partial \boldsymbol{\eta}}\big]$. Finally, the FIM of the unknown location parameters is given by $\mathbf{J}_{\boldsymbol{\eta}}=\mathbf{T}_\mathrm{tot}\mathbf{J}_{\boldsymbol{\gamma},\mathrm{tot}}\mathbf{T}_\mathrm{tot}^\mathrm{T}$.

\section{RIS Beamforming Design}
For a single satellite-RIS pair, this section proposes an optimal RIS beamforming design. 
Since there exists no common CPU to control and coordinate between the RIS and satellites, the joint design of the satellite precoders and RIS phase profiles (as reported in~\cite{Henk2022Beamforming}) is infeasible. To this end, we adopt random precoders for the LEO satellite and optimize the RIS profile $\boldsymbol{\omega}$ to achieve optimal localization performance. Since we fixed the RIS profile over transmissions, the codebook solution for RIS profile design in~\cite{Henk2022Beamforming} is also not applicable, motivating a distinct solution for LEO satellite localization.

Define matrix $\mathbf{B}=\left[\mathbf{b}_{\mathrm{r}}^*(\boldsymbol{\phi}_\mathrm{ru}),\frac{\partial \mathbf{b}^*_\mathrm{r}\left(\boldsymbol{\phi}_\mathrm{ru}\right)}{\partial \phi_\mathrm{ru}^\mathrm{az}},\frac{\partial \mathbf{b}^*_\mathrm{r}\left(\boldsymbol{\phi}_\mathrm{ru}\right)}{\partial \phi_\mathrm{ru}^\mathrm{el}}\right]$. For any $\boldsymbol{\omega}$, we can perform an orthogonal decomposition as
\begin{equation}
    \boldsymbol{\omega} = \mathbf{\Pi}_\mathbf{B}\boldsymbol{\omega} + \mathbf{\Pi}_\mathbf{B}^\perp\boldsymbol{\omega},
\end{equation}
where $\mathbf{\Pi}_\mathbf{B}\triangleq\mathbf{B}(\mathbf{B}^\mathrm{H}\mathbf{B})^{-1}\mathbf{B}^\mathrm{H}$ denotes the orthogonal projector onto the column space of $\mathbf{B}$ and $\mathbf{\Pi}_\mathbf{B}^\perp\triangleq\mathbf{I}-\mathbf{\Pi}_\mathbf{B}$. Following a similar routine as~\cite{Henk2022Beamforming}, we present the following remark and proposition.
\begin{remark}
    The FIM $\mathbf{J}_{\boldsymbol{\gamma}}$ in~\eqref{eq:S-B} does not depend on the component $\mathbf{\Pi}_\mathbf{B}^\perp\boldsymbol{\omega}$. This can be verified by observing the expressions of $\frac{\partial \mu_m[n]}{\partial \boldsymbol{\gamma}}$ in Section~\ref{sec:singleLEOCRB}, where the dependence of $\mathbf{J}_{\boldsymbol{\gamma}}$ on $\boldsymbol{\omega}$ is only through the elements of $\mathbf{B}^\mathrm{H}\boldsymbol{\omega}$.
\end{remark}
\begin{proposition}
    Under the constraint $\|\boldsymbol{\omega}\|_\infty\leq 1$ in~\eqref{eq:omegaconstraint}, the optimal $\boldsymbol{\omega}$ that minimizes PEB must satisfy $\|\mathbf{\Pi}_\mathbf{B}\boldsymbol{\omega}\|_\infty=1$.
\end{proposition}
\begin{proof}
We prove Proposition~1 via contradiction. Suppose there is an optimal $\boldsymbol{\omega}^\divideontimes=\mathbf{\Pi}_\mathbf{B}\boldsymbol{\omega}^\divideontimes + \mathbf{\Pi}_\mathbf{B}^\perp\boldsymbol{\omega}^\divideontimes$, $\|\boldsymbol{\omega}^\divideontimes\|_\infty\leq 1$, and $\|\mathbf{\Pi}_\mathbf{B}\boldsymbol{\omega}^\divideontimes\|_\infty<1$. Based on Remark~1, we have 
\begin{equation}
\resizebox{.55\hsize}{!}{$
\begin{aligned}
\mathbf{b}_{\mathrm{r}}^\mathrm{T}(\boldsymbol{\phi}_\mathrm{ru})\boldsymbol{\omega}^\divideontimes &= \mathbf{b}_{\mathrm{r}}^\mathrm{T}(\boldsymbol{\phi}_\mathrm{ru})\mathbf{\Pi}_\mathbf{B}\boldsymbol{\omega}^\divideontimes,\\
\frac{\partial \mathbf{b}^\mathrm{T}_\mathrm{r}\left(\boldsymbol{\phi}_\mathrm{ru}\right)}{\partial \phi_\mathrm{ru}^\mathrm{az}}\boldsymbol{\omega}^\divideontimes &= \frac{\partial \mathbf{b}^\mathrm{T}_\mathrm{r}\left(\boldsymbol{\phi}_\mathrm{ru}\right)}{\partial \phi_\mathrm{ru}^\mathrm{az}}\mathbf{\Pi}_\mathbf{B}\boldsymbol{\omega}^\divideontimes,\\
\frac{\partial \mathbf{b}^\mathrm{T}_\mathrm{r}\left(\boldsymbol{\phi}_\mathrm{ru}\right)}{\partial \phi_\mathrm{ru}^\mathrm{el}}\boldsymbol{\omega}^\divideontimes &= \frac{\partial \mathbf{b}^\mathrm{T}_\mathrm{r}\left(\boldsymbol{\phi}_\mathrm{ru}\right)}{\partial \phi_\mathrm{ru}^\mathrm{el}}\mathbf{\Pi}_\mathbf{B}\boldsymbol{\omega}^\divideontimes.
\end{aligned}
$}
\end{equation}
Then we can construct an alternative solution as $\boldsymbol{\omega}^\circledast=\frac{1}{\|\mathbf{\Pi}_\mathbf{B}\boldsymbol{\omega}^\divideontimes\|_\infty} \mathbf{\Pi}_\mathbf{B}\boldsymbol{\omega}^\divideontimes$ with $\|\boldsymbol{\omega}^\circledast\|_\infty = 1$ that satisfies the constraint in~\eqref{eq:omegaconstraint}, which gives
\begin{equation}
\resizebox{.99\hsize}{!}{$
\begin{aligned}
\mathbf{b}_{\mathrm{r}}^\mathrm{T}(\boldsymbol{\phi}_\mathrm{ru})\boldsymbol{\omega}^\circledast &= \mathbf{b}_{\mathrm{r}}^\mathrm{T}(\boldsymbol{\phi}_\mathrm{ru})\mathbf{\Pi}_\mathbf{B}\boldsymbol{\omega}^\circledast=\frac{1}{\|\mathbf{\Pi}_\mathbf{B}\boldsymbol{\omega}^\divideontimes\|_\infty}\mathbf{b}_{\mathrm{r}}^\mathrm{T}(\boldsymbol{\phi}_\mathrm{ru})\mathbf{\Pi}_\mathbf{B}\boldsymbol{\omega}^\divideontimes,\nonumber\\
\frac{\partial \mathbf{b}^\mathrm{T}_\mathrm{r}\left(\boldsymbol{\phi}_\mathrm{ru}\right)}{\partial \phi_\mathrm{ru}^\mathrm{az}}\boldsymbol{\omega}^\circledast &= \frac{\partial \mathbf{b}^\mathrm{T}_\mathrm{r}\left(\boldsymbol{\phi}_\mathrm{ru}\right)}{\partial \phi_\mathrm{ru}^\mathrm{az}}\mathbf{\Pi}_\mathbf{B}\boldsymbol{\omega}^\circledast=\frac{1}{\|\mathbf{\Pi}_\mathbf{B}\boldsymbol{\omega}^\divideontimes\|_\infty}\frac{\partial \mathbf{b}^\mathrm{T}_\mathrm{r}\left(\boldsymbol{\phi}_\mathrm{ru}\right)}{\partial \phi_\mathrm{ru}^\mathrm{az}}\mathbf{\Pi}_\mathbf{B}\boldsymbol{\omega}^\divideontimes,\nonumber \\
\frac{\partial \mathbf{b}^\mathrm{T}_\mathrm{r}\left(\boldsymbol{\phi}_\mathrm{ru}\right)}{\partial \phi_\mathrm{ru}^\mathrm{el}}\boldsymbol{\omega}^\circledast &= \frac{\partial \mathbf{b}^\mathrm{T}_\mathrm{r}\left(\boldsymbol{\phi}_\mathrm{ru}\right)}{\partial \phi_\mathrm{ru}^\mathrm{el}}\mathbf{\Pi}_\mathbf{B}\boldsymbol{\omega}^\circledast=\frac{1}{\|\mathbf{\Pi}_\mathbf{B}\boldsymbol{\omega}^\divideontimes\|_\infty}\frac{\partial \mathbf{b}^\mathrm{T}_\mathrm{r}\left(\boldsymbol{\phi}_\mathrm{ru}\right)}{\partial \phi_\mathrm{ru}^\mathrm{el}}\mathbf{\Pi}_\mathbf{B}\boldsymbol{\omega}^\divideontimes.\nonumber
\end{aligned}
$}
\end{equation}
Since $\frac{1}{\|\mathbf{\Pi}_\mathbf{B}\boldsymbol{\omega}^\divideontimes\|_\infty}>1$, $\boldsymbol{\omega}^\circledast$ generates a lower PEB than $\boldsymbol{\omega}^\divideontimes$ (due to scaling of $\mathbf{J}_{\boldsymbol{\gamma}}$), thus $\boldsymbol{\omega}^\divideontimes$ cannot be the optimal solution, which completes the proof. 
\end{proof} 

Based on Remark~1 and Proposition~1, the optimal RIS profile lies in the column space of $\mathbf{B}$ and can be formulated~as
\begin{align}\label{eq:minPEB}
    \min_{c_1,c_2,c_3}&\quad \text{PEB}(\boldsymbol{\omega})\\
    \text{s.t.}&\quad \boldsymbol{\omega}=c_1\mathbf{b}_{\mathrm{r}}^*(\boldsymbol{\phi}_\mathrm{ru})+c_2\frac{\partial \mathbf{b}^*_\mathrm{r}\left(\boldsymbol{\phi}_\mathrm{ru}\right)}{\partial \phi_\mathrm{ru}^\mathrm{az}}+c_3 \frac{\partial \mathbf{b}^*_\mathrm{r}\left(\boldsymbol{\phi}_\mathrm{ru}\right)}{\partial \phi_\mathrm{ru}^\mathrm{el}},\nonumber\\
    &\quad\|\boldsymbol{\omega}\|_\infty = 1,\nonumber
\end{align}
where $(c_1, c_2, c_3)$ indicates the coordinate of $\boldsymbol{\omega}$ in the column space of $\mathbf{B}$, and~\eqref{eq:minPEB} can be solved by adopting a grid search procedure. Note that solving~\eqref{eq:minPEB} requires knowledge of $\boldsymbol{\phi}_\mathrm{ru}$.
In practice, we can assume that a rough estimate of $\mathbf{p}_{\mathrm{u}}$ is available, using which a rough $\boldsymbol{\phi}_\mathrm{ru}$ can be obtained.

\section{Numerical Results}

This section presents the simulation results to evaluate the performance of the RIS-aided LEO satellite localization. 
We consider a scenario comprising a UE located at an unknown position $\mathbf{p}_{\mathrm{u}}=[0,0,0]^\mathrm{T}$~m, $L$ RISs placed at known positions $\mathbf{p}_{\mathrm{r},\ell}$ ($\ell=1,\dots,L$), and $K$ LEO satellites at known locations $\mathbf{p}_{\mathrm{s},k}$ ($k=1,\ldots,K$) moving with a velocity $\boldsymbol{v}=[5.5,5.5,0]^\mathrm{T}$~km/s. We set  $\mathbf{p}_{\mathrm{r},\ell}=\mathbf{p}_{\mathrm{r},0}+(\ell-1)\mathbf{d}_\mathrm{r}$, $\mathbf{p}_{\mathrm{s},k}=\mathbf{p}_{\mathrm{s},0}+(k-1)\mathbf{d}_\mathrm{s}$, where $\mathbf{p}_{\mathrm{r},0}=[60,10,30]^\mathrm{T}$~m, $\mathbf{p}_{\mathrm{s},0}=[-100,100,550]^\mathrm{T}$~km, $\mathbf{d}_\mathrm{r}=[0,20,0]^\mathrm{T}$~m, and $\mathbf{d}_\mathrm{s}=[-30,30,-5]^\mathrm{T}$~km. When simulating the terrestrial BS-based localization, the LEO satellites are replaced by a BS placed at $\mathbf{p}_{\mathrm{b}}=[-100,100,50]^\mathrm{T}$~m with zero velocity.
Each of the LEO satellites and BS is equipped with a $2\times 2$ array of antenna. We consider $M=128$ transmissions of pilot signals with transmission period $T=10$~ms. OFDM modulation with $N=3300$ subcarriers centered around a carrier frequency $f_c=12.7$~GHz and a bandwidth $B=240$~MHz is used. The clock offset is set as $\Delta=100$~ns.

We first compare the proposed beamforming design with the existing alternatives, i.e., random and directional beamforming~\cite{Henk2021directional}. Using a single LEO satellite, the comparisons are conducted by varying the number of RISs and their UPAs' size. As shown in Fig.~\ref{Fig_Biasvsd_1}, the PEB based on all three beamforming schemes decreases as the RIS size increases. Besides, the more RISs we deploy, the lower the PEB. Under all tested scenarios, the proposed beamforming design consistently outperforms the other two benchmarks, delivering considerably lower PEB. This reveals that the proposed RIS beamforming design can effectively improve LEO satellite-based localization accuracy.

Fig.~\ref{Fig_Biasvsd_2} demonstrates the PEB evaluation results using a terrestrial BS and $K$ (varying from 1 to 17) LEO satellites based on random, directional, and proposed beamforming designs. Here the number of RISs is fixed as 3, each equipped with $N_\mathrm{r}=10 \times 10$ elements. It is clearly shown that the proposed beamforming design can offer better localization performance than the other two beamforming schemes. The terrestrial BS is used as a benchmark to compare the localization performance based on terrestrial and non-terrestrial networks. Under the considered configurations, the performance of the RIS-aided LEO satellite localization cannot reach a comparable level to that of terrestrial localization. However, the performance of the proposed localization paradigm improves as the number of available LEO satellites increases, delivering meter- or even sub-meter-level localization accuracy. Given the abundant satellite resources, the results in Fig.~\ref{Fig_Biasvsd_2} reveal the potential of LEO satellite localization.

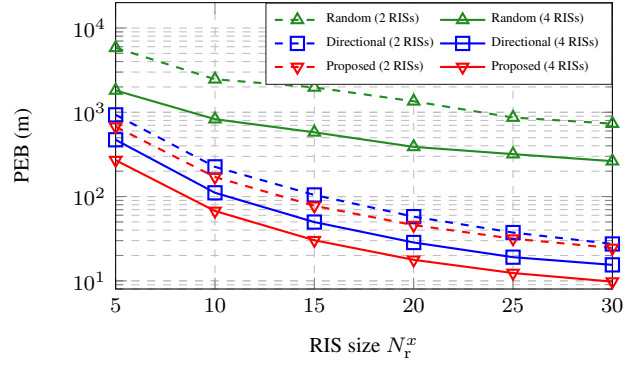
\begin{figure}[t]
    \centering
    \definecolor{ForestGreen}{rgb}{0.1333    0.5451    0.1333}
	\begin{tikzpicture}
\begin{axis}[%
width=2.6in, height=1.5in, at={(0,0)}, scale only axis,
xmin=5, xmax=30, 
xlabel style={font=\color{white!15!black},font=\footnotesize},
xticklabel style = {font=\color{white!15!black},font=\footnotesize},
xlabel={RIS size $N_\mathrm{r}^x$},
ymode=log,
ymin=8, ymax=20000,
yminorticks=true,
ylabel style={font=\color{white!15!black},font=\footnotesize},
yticklabel style = {font=\color{white!15!black},font=\footnotesize},
ylabel={PEB (m)},
axis background/.style={fill=white},
xmajorgrids,
ymajorgrids,
yminorgrids,
grid style={dashed},
legend style={at={(1,1)}, anchor=north east, legend cell align=left, align=left, font=\tiny, legend columns=2, draw=white!15!black}
]
\addplot [color=ForestGreen, line width=0.8pt, dashed, mark=triangle, mark options={solid, ForestGreen}, mark size=2.5pt]
  table[row sep=crcr]{%
5	5819.44904087368\\
10	2465.40817664247\\
15	1974.16648831318\\
20	1350.91849317878\\
25	866.509461910742\\
30	733.074652199118\\
};
\addlegendentry{Random (2 RISs)}

\addplot [color=ForestGreen, line width=0.8pt, mark=triangle, mark options={solid, ForestGreen}, mark size=2.5pt]
  table[row sep=crcr]{%
5	1816.56324013215\\
10	828.613709367751\\
15	578.451716011278\\
20	388.68112486249\\
25	318.80605571032\\
30	263.183229158896\\
};
\addlegendentry{Random (4 RISs)}

\addplot [color=blue, line width=0.8pt, dashed, mark=square, mark options={solid, blue}, mark size=2.5pt]
  table[row sep=crcr]{%
5	932.347518140475\\
10	225.107621232917\\
15	104.470684064286\\
20	57.7615301129481\\
25	37.3955070258747\\
30	27.396608627684\\
};
\addlegendentry{Directional (2 RISs)}

\addplot [color=blue, line width=0.8pt, mark=square, mark options={solid, blue}, mark size=2.5pt]
  table[row sep=crcr]{%
5	471.894213786977\\
10	110.858126719359\\
15	49.8163449858334\\
20	28.5192269680778\\
25	19.0921849622307\\
30	15.47611429105\\
};
\addlegendentry{Directional (4 RISs)}

\addplot [color=red, line width=0.8pt, dashed, mark=triangle, mark options={solid, rotate=180, red}, mark size=2.5pt]
  table[row sep=crcr]{%
5	674.693139065936\\
10	170.610853544582\\
15	77.7413340964761\\
20	45.8109730324483\\
25	31.6494164572451\\
30	24.5404194456621\\
};
\addlegendentry{Proposed (2 RISs)}

\addplot [color=red, line width=0.8pt, mark=triangle, mark options={solid, rotate=180, red}, mark size=2.5pt]
  table[row sep=crcr]{%
5	270.271055363533\\
10	67.4634939071008\\
15	30.3166984535957\\
20	17.759329326992\\
25	12.3743202333877\\
30	9.78628165183474\\
};
\addlegendentry{Proposed (4 RISs)}

\end{axis}
\end{tikzpicture}  
\caption{ PEB versus RIS size $N_\mathrm{r}^x$ with $N_\mathrm{r}^y=N_\mathrm{r}^x$.
  }
  \label{Fig_Biasvsd_1}
\end{figure}

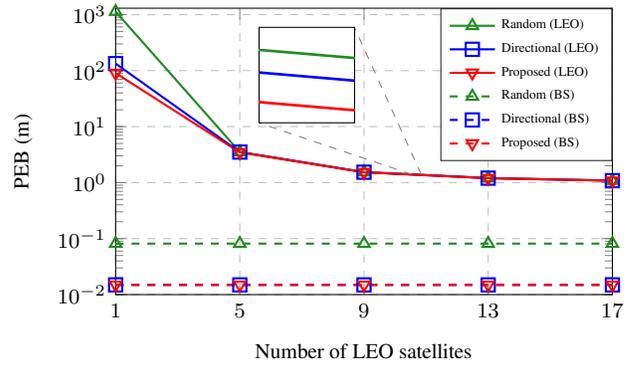
\begin{figure}[t]
    \centering
    \definecolor{ForestGreen}{rgb}{0.1333    0.5451    0.1333}
\begin{tikzpicture}
\begin{axis}[%
width=2.6in, height=1.5in, at={(0,0)}, scale only axis,
xmin=1, xmax=17,
xlabel style={font=\color{white!15!black},font=\footnotesize},
xticklabel style = {font=\color{white!15!black},font=\footnotesize},
xlabel={Number of LEO satellites},
xtick={1,5,9,13,17},
ymode=log,
ymin=0.01, ymax=1300,
yminorticks=true,
ylabel style={font=\color{white!15!black},font=\footnotesize},
yticklabel style = {font=\color{white!15!black},font=\footnotesize},
ytick={1e-2,1e-1,1,1e1,1e2,1e3,1e4},
ylabel={PEB (m)},
axis background/.style={fill=white},
xmajorgrids,
ymajorgrids,
grid style={dashed},
legend style={at={(1,1)}, anchor=north east, legend cell align=left, align=left, font=\tiny, draw=white!15!black}
]
\addplot [color=ForestGreen, line width=0.8pt, mark=triangle, mark options={solid, ForestGreen}, mark size=2.5pt]
  table[row sep=crcr]{%
1	1131.95879660419\\
5	3.51387434829026\\
9	1.53846738238373\\
13	1.20638074801396\\
17	1.08749469450172\\
};
\addlegendentry{Random (LEO)}

\addplot [color=blue, line width=0.8pt, mark=square, mark options={solid, blue}, mark size=2.5pt]
  table[row sep=crcr]{%
1	133.499180189332\\
5	3.50868601246847\\
9	1.53433630703868\\
13	1.20110144389551\\
17	1.08291646942927\\
};
\addlegendentry{Directional (LEO)}

\addplot [color=red, line width=0.8pt, mark=triangle, mark options={solid, rotate=180, red}, mark size=2.5pt]
  table[row sep=crcr]{%
1	90.3014594988886\\
5	3.46424125846131\\
9	1.52729609481913\\
13	1.19598635700331\\
17	1.07765055975128\\
};
\addlegendentry{Proposed (LEO)}

\addplot [color=ForestGreen, line width=0.8pt, dashed, mark=triangle, mark options={solid, ForestGreen}, mark size=2.5pt]
  table[row sep=crcr]{%
1	0.0813705387726959\\
5	0.0813705387726959\\
9	0.0813705387726959\\
13	0.0813705387726959\\
17	0.0813705387726959\\
};
\addlegendentry{Random (BS)}

\addplot [color=blue, line width=0.8pt, dashed, mark=square, mark options={solid, blue}, mark size=2.5pt]
  table[row sep=crcr]{%
1	0.0148979434471336\\
5	0.0148979434471336\\
9	0.0148979434471336\\
13	0.0148979434471336\\
17	0.0148979434471336\\
};
\addlegendentry{Directional (BS)}

\addplot [color=red, line width=0.8pt, dashed, mark=triangle, mark options={solid, rotate=180, red}, mark size=2.5pt]
  table[row sep=crcr]{%
1	0.0148913817863124\\
5	0.0148913817863124\\
9	0.0148913817863124\\
13	0.0148913817863124\\
17	0.0148913817863124\\
};
\addlegendentry{Proposed (BS)}
\end{axis}

\draw[gray,dashed] (1.55in,0.63in) -- (0.75in,0.9in);
\draw[gray,dashed] (1.6in,0.63in) -- (1.25in,1.4in);

\begin{axis}[%
width=0.5in, height=0.5in,
at={(0.75in,0.9in)},
scale only axis,
xmin=10.49,xmax=10.51,ymin=1.39,ymax=1.41,ymode=log,
axis background/.style={fill=white},
xtick={\empty},
xticklabels={\empty},
ytick={\empty},
yticklabels={\empty}
]
\addplot [color=ForestGreen, line width=1pt, mark=triangle, mark options={solid, ForestGreen}]
  table[row sep=crcr]{%
9	1.53846738238373\\
13	1.20638074801396\\
};

\addplot [color=blue, line width=1pt, mark=square, mark options={solid, blue}]
  table[row sep=crcr]{%
9	1.53433630703868\\
13	1.20110144389551\\
};

\addplot [color=red, line width=1pt, mark=triangle, mark options={solid, rotate=180, red}]
  table[row sep=crcr]{%
9	1.52729609481913\\
13	1.19598635700331\\
};
\end{axis}	
\end{tikzpicture}%
  \caption{ PEB versus number of LEO satellites. 
  }
  \label{Fig_Biasvsd_2}
\end{figure}

\section{Conclusion}
This work investigated the performance of RIS-aided localization using LEO satellite signals. Specifically, we derived the fundamental PEB, based on which we further proposed an optimal RIS beamforming design that minimizes the derived PEB. Numerical results demonstrated the efficacy of the proposed localization-oriented beamforming design and revealed the promise of the synergy between LEO satellites and RISs for localization with meter- or even sub-meter-level accuracy. 

\bibliographystyle{IEEEbib}
\bibliography{references}

\begin{thebibliography}{10}

\bibitem{Li2022LEO}
K.~X. Li, L.~You, J.~Wang, X.~Gao, C.~G. Tsinos, S.~Chatzinotas, and
  B.~Ottersten,
\newblock ``Downlink transmit design for massive {MIMO LEO} satellite
  communications,''
\newblock {\em IEEE Trans. Commun.}, vol. 70, no. 2, pp. 1014--1028, 2022.

\bibitem{kassas2019new}
Z.~Kassas, J.~Morales, and J.~Khalife,
\newblock ``New-age satellite-based navigation--{STAN}: simultaneous tracking
  and navigation with {LEO} satellite signals,''
\newblock {\em Inside GNSS Magazine}, vol. 14, no. 4, pp. 56--65, 2019.

\bibitem{reid2018broadband}
T.~G. Reid, A.~M. Neish, T.~Walter, and P.~K. Enge,
\newblock ``Broadband {LEO} constellations for navigation,''
\newblock {\em NAVIGATION: Journal of the Institute of Navigation}, vol. 65,
  no. 2, pp. 205--220, 2018.

\bibitem{Kozhaya2023Bilnd}
S.~Kozhaya, H.~Kanj, and Z.~M. Kassas,
\newblock ``Multi-constellation blind beacon estimation, {Doppler} tracking,
  and opportunistic positioning with {OneWeb}, {Starlink}, {Iridium NEXT}, and
  {Orbcomm} {LEO} satellites,''
\newblock in {\em Proc. IEEE/ION Position, Location and Navigation Symposium
  (PLANS)}, Monterey, CA, USA, Apr. 2023, pp. 1184--1195.

\bibitem{dureppagari2023ntn}
H.~K. Dureppagari, C.~Saha, H.~S. Dhillon, and R.~M. Buehrer,
\newblock ``{NTN}-based {6G} localization: Vision, role of {LEOs}, and open
  problems,''
\newblock {\em arXiv preprint arXiv:2305.12259}, 2023.

\bibitem{Zheng20235G}
P.~Zheng, X.~Liu, T.~Ballal, and T.~Y. Al-Naffouri,
\newblock ``{5G}-aided {RTK} positioning in {GNSS}-deprived environments,''
\newblock in {\em 31th European Signal Processing Conference (EUSIPCO)}, 2023.

\bibitem{Zheng2023Attitude}
P.~Zheng, X.~Liu, T.~Ballal, and T.~Y. Al-Naffouri,
\newblock ``Attitude determination in urban canyons: A synergy between {GNSS}
  and {5G} observations,''
\newblock {\em ION GNSS+}, 2023.

\bibitem{Bjornson2022Reconfigurable}
E.~Björnson, H.~Wymeersch, B.~Matthiesen, P.~Popovski, L.~Sanguinetti, and
  E.~de~Carvalho,
\newblock ``Reconfigurable intelligent surfaces: A signal processing
  perspective with wireless applications,''
\newblock {\em IEEE Signal Process. Mag.}, vol. 39, no. 2, pp. 135--158, 2022.

\bibitem{Tekb2022Reconfigurable}
K.~Tekbıyık, G.~K. Kurt, A.~R. Ekti, and H.~Yanikomeroglu,
\newblock ``Reconfigurable intelligent surfaces in action for nonterrestrial
  networks,''
\newblock {\em IEEE Veh. Technol. Mag.}, vol. 17, no. 3, pp. 45--53, 2022.

\bibitem{Henk2022Beamforming}
A.~Fascista, M.~F. Keskin, A.~Coluccia, H.~Wymeersch, and G.~Seco-Granados,
\newblock ``{RIS}-aided joint localization and synchronization with a
  single-antenna receiver: Beamforming design and low-complexity estimation,''
\newblock {\em IEEE J. Sel. Top. Signal Process.}, vol. 16, no. 5, pp.
  1141--1156, 2022.

\bibitem{Zheng2023JrCUP}
P.~Zheng, H.~Chen, T.~Ballal, M.~Valkama, H.~Wymeersch, and T.~Y. Al-Naffouri,
\newblock ``{JrCUP}: Joint {RIS} calibration and user positioning for {6G}
  wireless systems,''
\newblock {\em preprint arXiv:2304.00631}, 2023.

\bibitem{Chen2023Multi}
H.~Chen, P.~Zheng, M.~F. Keskin, T.~Y. Al-Naffouri, and H.~Wymeersch,
\newblock ``{Multi-RIS-enabled 3D} sidelink positioning,''
\newblock {\em preprint arXiv:2302.12459}, 2023.

\bibitem{Rui2022LEO}
B.~Zheng, S.~Lin, and R.~Zhang,
\newblock ``Intelligent reflecting surface-aided {LEO} satellite communication:
  Cooperative passive beamforming and distributed channel estimation,''
\newblock {\em IEEE J. Sel. Areas Commun.}, vol. 40, no. 10, pp. 3057--3070,
  2022.

\bibitem{Shen2022RIS}
S.~Shen, B.~Clerckx, and R.~Murch,
\newblock ``Modeling and architecture design of reconfigurable intelligent
  surfaces using scattering parameter network analysis,''
\newblock {\em IEEE Trans. Wireless Commun.}, vol. 21, no. 2, pp. 1229--1243,
  2022.

\bibitem{You2020Model}
L.~You, K.~X. Li, J.~Wang, X.~Gao, X.~Xia, and B.~Ottersten,
\newblock ``Massive {MIMO} transmission for {LEO} satellite communications,''
\newblock {\em IEEE J. Sel. Areas Commun.}, vol. 38, no. 8, pp. 1851--1865,
  2020.

\bibitem{Zheng2022Coverage}
P.~Zheng, T.~Ballal, H.~Chen, H.~Wymeersch, and T.~Y. Al-Naffouri,
\newblock ``Coverage analysis of joint localization and communication in {THz}
  systems with {3D} arrays,''
\newblock {\em TechRxiv preprint}, 2022.

\bibitem{Neinavaie2023OFDM}
M.~Neinavaie and Z.~M. Kassas,
\newblock ``Unveiling {Starlink} {LEO} satellite {OFDM}-like signal structure
  enabling precise positioning,''
\newblock {\em IEEE Trans. Aerosp. Electron. Syst.}, pp. 1--4, 2023.

\bibitem{Ayach2014Spatially}
O.~E. Ayach, S.~Rajagopal, S.~Abu-Surra, Z.~Pi, and R.~W. Heath,
\newblock ``Spatially sparse precoding in millimeter wave {MIMO} systems,''
\newblock {\em IEEE Trans. Wireless Commun.}, vol. 13, no. 3, pp. 1499--1513,
  2014.

\bibitem{Kay1993Slepian}
S.~M. Kay,
\newblock {\em Fundamentals of Statistical Signal Processing: Estimation
  Theory},
\newblock Englewood Cliffs, NJ, USA: Prentice-Hall, Inc., 1993.

\bibitem{Henk2021directional}
A.~Kakkavas, H.~Wymeersch, G.~Seco-Granados, M.~H.~Castañeda G., R.~A.
  Stirling-Gallacher, and J.~A. Nossek,
\newblock ``Power allocation and parameter estimation for multipath-based {5G}
  positioning,''
\newblock {\em IEEE Trans. Wireless Commun.}, vol. 20, no. 11, pp. 7302--7316,
  2021.

\end{thebibliography}

\end{document}